\documentclass[12pt,a4paper,twoside]{article}
\title{A Quantum Mechanical Bound for Space-Energy Cost with Respect to the Von Neumann Entropy}
\author{Christoph Haupt}
	\usepackage{amsmath, amsfonts, amssymb, amsthm, bbm, color, braket}
	\usepackage{listings}
	\usepackage[nottoc]{tocbibind}
	\usepackage[normalem]{ulem}	
	\usepackage{chngcntr}		
	\usepackage[ngerman,english]{babel}
	\usepackage[utf8]{inputenc}
	\usepackage[]{graphicx}
	\usepackage{epstopdf}
	\usepackage{hyperref}			

\renewcommand{\S}{\vNentropy}
\newcommand{\C}{\mathbb{C}}
\newcommand{\N}{\mathbb{N}}
\newcommand{\R}{\mathbb{R}}
\newcommand{\Cn}{\mathbb{C}^{n\times n}}

\newcommand{\1}{\mathbbm{1}}
\newcommand{\h}{\hbar}
\newcommand{\dx}{\,\mathrm{d}}
\newcommand{\VERT}[1]{{\left\vert\kern-0.25ex\left\vert\kern-0.25ex\left\vert #1 
    \right\vert\kern-0.25ex\right\vert\kern-0.25ex\right\vert}}

\newcommand{\tf}{\tilde{f}}
\newcommand{\eps}{\epsilon}

\newcommand{\psid}{\ket{\psi_{\vec{n}}}}
\newcommand{\ketbra}[1]{| #1 \rangle\langle #1 |}
\newcommand{\ketbras}[2]{| #1 \rangle\langle #2 |}
\newcommand{\limp}{\lim_{p\to0^+}}
\newcommand{\rhob}{\rho_{\beta_1,\beta_2}}

\newcommand{\ksumi}{\sum_{k=0}^\infty}
\newcommand{\ddt}{\frac{d}{dt}}

\newtheoremstyle{dotless}
	{10pt}
	{10pt}
	{\itshape}
	{}
	{\bfseries}
	{}
	{.5em}
	{}
\newtheoremstyle{dot}
	{10pt}
	{10pt}
	{\itshape}
	{}
	{\bfseries}
	{.}
	{.5em}
	{}
\newtheoremstyle{dotless_def}
	{10pt}
	{10pt}
	{}
	{}
	{\bfseries}
	{}
	{.5em}
	{}
\newtheoremstyle{dot_def}
	{10pt}
	{10pt}
	{}
	{}
	{\bfseries}
	{.}
	{.5em}
	{}
\newtheoremstyle{todo_note}
	{9pt}
	{9pt}
	{\color{red}}
	{}
	{\bfseries \color{red}}
	{:}
	{.5em}
	{}
\theoremstyle{dot}
	\newtheorem{theorem}{Theorem}

	\newtheorem{lemma}[theorem]{Lemma}
\theoremstyle{dot_def}
	
	\newtheorem*{definition*}{Definition}
	\newtheorem*{bemerkung*}{Bemerkung}
	\newtheorem*{remark*}{Remark}
\theoremstyle{todo_note}
	\newtheorem{ToDo}{ToDo}

\DeclareMathOperator{\sgn}{sgn}
\DeclareMathOperator{\vNentropy}{S}

\DeclareMathOperator{\tr}{tr}
\DeclareMathOperator{\spa}{span}
\DeclareMathOperator{\var}{var}
\DeclareMathOperator{\supp}{supp}
\DeclareMathOperator{\inter}{int}

\begin{document}
\showboxdepth=\maxdimen
\showboxbreadth=\maxdimen
\pagestyle{empty}
\begin{titlepage}
\begin{center}
\includegraphics{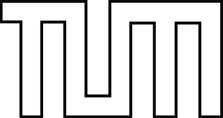}\\[3mm]
\sf
{\Large
  Technische Universit\"at M\"unchen\\[5mm]
  Department of Mathematics\\[8mm]
}
\normalsize
\includegraphics{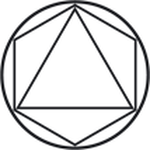}\\[15mm]

Master's Thesis\\[15mm]

{\Huge
A Quantum Mechanical Bound for Space-Energy Cost with Respect to the Von Neumann Entropy
}
\bigskip

\normalsize

Christoph Haupt
\end{center}
\vspace{75mm}

Supervisor: Prof. Dr. Michael M. Wolf 
\medskip

Advisor: Martin Idel
\medskip

Submission Date: 15.09.2015

\end{titlepage}
\tableofcontents
\newpage
\pagenumbering{arabic}
\pagestyle{headings}
\section{Introduction}
This thesis discusses the possibility of uncertainty relations for space and energy given a state of fixed entropy. In particular, it discusses the results in \cite{dam-nguyen}. There, the authors propose a lower bound for the mixed cost in energy and space required for physically storing information in a quantum mechanical system.

We first critically examine the justifications for the bound given in the paper.
This is done from a mathematical point of view, in contrast to the more physically motivated original paper.
Then we give two examples that illustrate the limitations of this inequality.

In section \ref{sec_GVP} we prove the variational principle for Gibbs states, which is a central theorem in this subject.

Using this, we present numerical results to find an alternative energy-space bound.
We do this in the finite dimensional version of this problem.
They indicate a slightly different version of the inequality.

In the end we describe a promising ansatz to find a lower energy-space bound that depends on the amount stored information.
Unfortunately that did not result in a satisfactory result.
\section{The Energy Surface Bound}\label{sec_bound}

In this chapter we examine the main ideas of \cite{dam-nguyen}.
After giving the basic definitions and introducing the physical background, we introduce the formula and give  a sketch of its proof.
In the end we present a collection of examples that illustrate the limitations for the inequality.

\subsection{Preliminaries}\label{sec_prel}

In \cite{dam-nguyen} the authors propose a lower bound for the mixed cost in energy and space required for physically storing information.
This is done in a non-relativistic quantum mechanical setting.
Therefore the storing device is given by its Hamiltonian $H$, which is acting on a subset of a separable, infinite dimensional Hilbert space $\mathcal{H}$.
We can assume without loss of generality that $\mathcal{H}=\mathcal{L}_2(\R^d)$.
$d$ is called the number of degrees of freedom of the system, but sometimes referred to as dimension as well.
We only consider Hamiltonians of the form 
\[
H=\frac{1}{2m}P^2+V(x),
\]
with the momentum operator $P:D(P)\to\mathcal{L}_2(\R^d)$, $D(P):=\{\psi\in \mathcal{L}_2(\R^d)|P(\psi)\in \mathcal{L}_2(\R^d)\}$,
\[
P(\psi)(x)=-i\hbar\sum_{i=1}^d\frac{\partial}{\partial x_i} \psi(x).
\]
$m$ is the mass of the system and $\h$ the reduced Planck constant. With this restriction we can model one or more particles in a device, but we cannot account for interaction between the particles.
It is physically reasonable to assume that $H$ has finite ground state energy, i.e. the spectrum is bounded from below. We even assume that the ground state energy is non-negative.
The set of bounded operators on $\mathcal{H}$ is denoted by $\mathcal{B(H)}$.
A particle (or multiple particles) in this device is described by a density matrix $\rho\in\mathcal{B(H)}$.
A density matrix or state is a positive trace-class operator with trace one, i.e. $\rho\geq0$, $\rho\in \mathcal{S}_1(\mathcal{H})$, $\tr(\rho)=1$.
We denote the set of all states with $\mathcal{S(H)}$.
To physically store information in this setting one has to prepare an energy eigenstate $\ketbra{\psi_i}$, where $\ket{\psi_i}$ is an eigenvector of $H$.
The quantity information is given by a probability distribution $(\lambda_i)_i\in l^1(\R)$, and the storage of this information in the device then corresponds to the preparation of the state
\[
\sum_i \lambda_i\ketbra{\psi_i}.
\]
We use the Von Neumann entropy $\S(\rho)$ as measure for the amount of information that is stored in a state $\rho$.
The Von Neumann entropy is the quantum mechanical analogue of the Shannon entropy. For more details on the Shannon entropy please consult \cite{tom-cover}.
The Von Neumann entropy is defined as $\S:\mathcal{S(H)}\to\R^+_0$,
\[
\S(\rho)=-\tr(\rho\log(\rho)).
\]
See lemma \ref{vnEntropy} for detailed information on the definition.
The cost of energy and space of a state $\rho$ are $\tr(\rho H)$ and 
\[
\var_\rho(Q)=\tr(\rho Q^2)-\tr(\rho Q)^2,
\]
respectively.
The one-dimensional position operators are defined by $Q_i:D(Q_i)\to L^2(\R^d)$,  $D(Q_i):=\{\psi\in \mathcal{L}_2(\R^d)|Q_i\psi\in \mathcal{L}_2(\R^d)\}$, $i\in\{1,\dotsc,d\}$,
\[
Q_i(\psi)(x)=x_i\psi(x)
\]
and the position operator in $d$ dimensions is $Q:D(Q)\to L^2(\R^d)$,  $D(Q):=\{\psi\in \mathcal{L}_2(\R^d)|Q\psi\in \mathcal{L}_2(\R^d)\}$,
\[
Q(\psi)(x)=\sum_{i=1}^d Q_i\psi(x).
\]

\subsection{The Statement}
The central statement in \cite{dam-nguyen} is an inequality similar to Heisenberg's famous uncertainty relation.
Here, the product of the variance of the observable position and the energy cost are lower bounded by a function of the entropy. The inequality in question is
\begin{equation}\label{dn_bound}
\tr\left(\rho H\right)  \tr(\rho Q^2) \geq \frac{\h^2d^2}{2m}(\exp(\S(\rho)/d)-1)^2
\end{equation}
for all $\rho\in\mathcal{S(H)}$ and all Hamiltonians $H$. $m$ is usually the mass of the particle(s) in the system.

We only consider Hamiltonians with ground state energy 0.
Thereby we prevent physically meaningless shifts in the Hamiltonian, while every Hamiltonian $H$ with positive ground state energy $\lambda_{min}$ can be normalized via $\tilde{H}=H-\lambda_{min} \1$, reducing its costs in the process.

The inequality \eqref{dn_bound} contains $\tr(\rho Q^2)$ instead of the actual spacial cost $\var_\rho(Q)$.
Although $\tr(\rho Q^2)$ is strictly larger for some $\rho$, we show in lemma \ref{lem_varq} that the two expressions are interchangeable in this context and result in the same statement.

\begin{lemma}\label{lem_varq}
Let $S>0$. Then we have
\begin{equation}\label{eq_cost_min}
\min_{\rho,H}\tr(\rho H)\tr(\rho Q^2) =\min_{\rho,H}\tr(\rho H)\var(\rho Q),
\end{equation}
where we optimize both times over all states $\rho$ which are diagonal in the energy eigenbasis of $H$ and have $\S(\rho)=S$, and all normalized Hamiltonians $H$.
\end{lemma}

\begin{proof}
Obviously we have $\geq$ in the proposed equality.
To show the reversed inequality, we fix $\rho$ and $H$ such that $\tr(\rho H)\tr(\rho Q^2)$ is minimal.
Such minimizing $\rho$ and $H$ exist, since they are taken from closed sets.
Let $(\psi_i)_i$ be the eigenbasis of $H$, $\rho=\sum_{i=0}^\infty\lambda_i\ketbra{\psi_i}$  and set $C:=\tr(\rho Q_k)$ for an arbitrary dimension $k\in\{1,\dotsc,d\}$.
We denote the $k$-th unit vector $e_k$ and define a shifted basis as
\[
\tilde{\psi}(x)=\psi(x+Ce_k),
\]
which is also a well defined orthonormal basis.
Analogously for $H=\frac{1}{2m}P^2+V$ we set the shifted Hamiltonian
\[
\tilde{H}=\frac{1}{2m}P^2+\tilde{V},
\]
where $\tilde{V}(x)=V(x+Ce_k)$. Then we have 
\begin{align*}
	\braket{\tilde{\psi}_i|\tilde{H}|\tilde{\psi}_i}
	&=\int_{\R^d} \frac{1}{2m}\overline{\tilde{\psi}_i(x)} P^2(\tilde{\psi}_i)(x) + \tilde{V}(x)|\tilde{\psi}_i(x)|^2 \dx x
	\\
	&=\int_{\R^d} \frac{1}{2m}\overline{{\psi}_i(x+Ce_k)} P^2({\psi}_i)(x+Ce_k) + {V}(x+Ce_k)|{\psi}_i(x+Ce_k)|^2 \dx x
	\\
	&=\int_{\R^d} \frac{1}{2m}\overline{{\psi}_i(x)} P^2({\psi}_i)(x) + {V}(x)|{\psi}_i(x)|^2 \dx x
	\\
	&=\braket{{\psi}_i|{H}|{\psi}_i}
\end{align*}
and
\begin{align*}
\braket{\tilde{\psi}_i|Q_k|\tilde{\psi}_i}&=
\int_{\R^d}  |\tilde{\psi}_i(x)|^2 x_k \dx x
\\
&=\int_{\R^d}  |{\psi}_i(x)|^2 (x_k-C) \dx x
\\
&=\braket{{\psi_i}|Q_k|{\psi_i}}-C.
\end{align*}
And we get for the position expectation value in dimension $k$ of $\tilde{\rho}$
\begin{align}\nonumber
\tr(\tilde{\rho}Q_k)&=\sum_{i=0}^\infty\lambda_i\braket{\tilde{\psi}_i|Q_k|\tilde{\psi}_i}
\\\nonumber
&=\sum_{i=0}^\infty\lambda_i\braket{{\psi}_i|Q_k|{\psi}_i}-\sum_{i=0}^\infty\lambda_i C
\\\nonumber
&=\tr({\rho}Q_k)-C
\\\label{eq_Q_k_0}
&=0
\end{align}
and the average energy
\begin{align*}
\tr(\tilde{\rho}\tilde{H})
&=\sum_{i=0}^\infty\lambda_i\braket{\tilde{\psi}_i|\tilde{H}|\tilde{\psi}_i}
\\
&=\tr({\rho}{H}).
\end{align*}
And we get for the expectation value of $Q_k^2$:
\begin{align*}
	\braket{\tilde{\psi}_i|Q_k^2|\tilde{\psi}_i}
	&=\int_{R^d}|\tilde{\psi}(x)|x_k^2\dx x
	\\
	&=\int_{R^d}|{\psi}(x)|(x_k^2-2Cx_k+C^2)\dx x
	\\
	&=\braket{{\psi}_i|Q_k^2|{\psi}_i}-C^2.
\end{align*}
We have already shown with \eqref{eq_Q_k_0} that $\tr(\tilde{\rho} Q_k^2)=\var_{\tilde{\rho}}(Q_k)$. Now we see that we haven't changed the variance at all:
\begin{align*}
	\tr(\tilde{\rho} Q_k^2)
	&=\sum_{i=0}^\infty\lambda_i\braket{\tilde{\psi}_i|Q_k^2|\tilde{\psi}_i}
	\\
	&=\sum_{i=0}^\infty\lambda_i(\braket{{\psi}_i|Q_k^2|{\psi}_i}-C^2)
	\\
	&=\var_\rho(Q_k).
\end{align*}
We can repeat this procedure for the new state $\tilde{\rho}$ and the new Hamiltonian $\tilde{H}$ for all dimensions $k\in\{1,\dotsc,d\}$. The resulting state $\tilde{\psi}$ in its Hamiltonian $\tilde{H}$ then has the same energy 
\[
\tr(\tilde{\rho}\tilde{H})=\tr({\rho}{H})
\]
and the same variance
\begin{align*}
\var_{\tilde{\rho}}(Q)&=\var_{{\rho}}(Q)
\\
&=\tr(\tilde{\rho} Q^2)
\end{align*}
and we have shown that the minimum of the left hand side in the proposed equation \eqref{eq_min_cost} is no less than the right hand side. Thus we have equality.
\end{proof}

\subsection{The Justification of the Bound}\label{ss_justific}

We will now summarize the arguments that led the authors of \cite{dam-nguyen} to the energy-surface bound \eqref{dn_bound}.
The authors first relax the condition that $\rho$ is diagonal in the eigenvectors of $H$ and fix the entropy $\S(\rho)=S>0$. Thus we prove \eqref{dn_bound} by calculating
\begin{equation}
\min_{\rho,H}\tr(\rho H)\tr(\rho Q^2) \label{eq_min_cost}
\end{equation}
over $\rho\in\mathcal{S(H)}$ with $\S(\rho)=S$ and all normalized Hamiltonians $H=\frac{1}{2m} P^2+V(x)$, with $V$ being an arbitrary potential in $\mathcal{H}$.

Now the authors make the pyhsically motivated conjecture that
\begin{equation} \label{eq_H_opt}
H_{opt}=\frac{1}{2m}P^2-\Big(1-\frac{d}{2}\Big)^2\frac{\hbar^2}{2m}Q^{-2}
\end{equation}
yields the optimal Hamiltonian in \eqref{eq_min_cost}, and we are left with the optimization problem
\[
\min_{\rho}\,\tr(\rho H_{opt}) \tr(\rho Q^2).
\]
This will be estimated using 
\begin{equation*}
\tilde{C}_\kappa:=\min_\rho\,\tr\rho\left(H_{opt}+\frac{\kappa}{2}r^2\right)
\end{equation*}
for $\kappa>0$.
By a Lemma, which we will look at later, we assume
\begin{equation}\label{eq_sum_cost}
\tilde{C}_\kappa\geq\hbar\sqrt{\frac{\kappa}{m}}d(\exp(S/d)-1).
\end{equation}
Plugging this in, we get
\begin{equation}
\tr \rho H_{opt}\geq\hbar\sqrt{\frac{\kappa}{m}}d(\exp(S/d)-1)-\frac{\kappa}{2}\tr\rho r^2
\end{equation}
for all $\rho\in\mathcal{S(H)}$, which yields
\begin{equation}
\tr \rho H_{opt}\tr\rho r^2\geq\sqrt{\kappa}\left(\hbar\frac{d}{\sqrt{m}}(\exp(\S/d)-1)-\frac{\kappa}{2}\tr\rho r^2\right)\tr\rho r^2.
\end{equation}
The right hand side is quadratic in $\sqrt{\kappa}$, so we can calculate its value at critical $\kappa$ and we get
\begin{equation}
\tr \rho H_{opt} \tr \rho r^2\geq\frac{\hbar^2}{2m}d^2(\exp(\S/d)-1)^2.
\end{equation}
This is the desired statement and all that's left to show is \eqref{eq_sum_cost}.
This is subject to a so called lemma in the last section of the paper \cite{dam-nguyen}. Though it is not a lemma in the mathematical sense since it lacks the necessary precision.

We consider a new Hamiltonian
\begin{equation*}
H=\frac12 P^2 -\frac{W}{q^2}+q^2
\end{equation*}
with $\vec{q}:=\vec{x}(m\kappa)^{1/4}$ and $W\in\R$.
$W$ is chosen such that
\begin{equation*}
\tilde{C}_\kappa=\hbar\sqrt{\kappa/m}\, \min_{\rho} \tr\rho H.
\end{equation*}
To find a minimal state $\rho$ we use Gibbs variational principle (the infinite dimensional version of theorem \ref{thm_gibbs} with $H_2=0$) and get
\begin{equation}\label{eq_proof_gibbs}
\rho=\frac{\exp(-\beta H)}{\|\exp(-\beta H)\|},
\end{equation}
for a $\beta\in\R$. This can be rewritten as
\[
\rho=\frac{1}{\|\exp(-\beta H)\|}\sum_i \exp(\beta E_i)\ketbra{\psi_i},
\]
where $\ket{\psi_i}$ are the eigenvectors and $E_i$ are the eigenenergies of $H$.
So we get
\begin{equation*}
\tilde{C}_\kappa=\frac{\hbar\sqrt{\kappa/m}}{\|\exp(-\beta H)\|}\sum_i \exp(\beta E_i)E_i.
\end{equation*}
By \cite{wolf74} the eigenstates have quantum numbers $n,l\in\N^+_0$ with energy
\[
E(n,l)=2n+\sqrt{l(l+d-2)},
\]
and degeneracy
\[
g(l)=\frac{(d+2l-2)(d+l-3)!}{l!(d-2)!}.
\]
So we get
\[
\tilde{C}_\kappa=\frac{\hbar\sqrt{\kappa/m}}{Z}\sum_{n,l} \exp(-\beta E(n,l))g(l)E(n,l),
\]
where
\begin{align*}
Z
&:=\|\exp(-\beta H)\|\\
&=\sum_{n,l}\exp(-\beta E(n,l))g(l)
\\
&=\underbrace{\sum_n(\exp(-2\beta n)}_{=:Z_n}\underbrace{\sum_l \exp(-\beta\sqrt{l(l+d-2))g(l)}}_{=:Z_l}.
\end{align*}
Please be aware that the subscripts $l$ and $n$ are not variables. Now we can write
\[
\tilde{C}_\kappa=\h\sqrt{\frac{\kappa}{m}}(U_n+U_l)
\]
for
\[
U_n:=\frac{1}{Z_n}\sum_n 2n\exp(-2\beta n)
\]
and
\[
U_l:=\frac{1}{Z_l}\sum_l\exp(-\beta\sqrt{l(l+d-2)})g(l)\sqrt{l(l+d-2}.
\]
Now by analytic transformations (where we assume $|\beta|<1$ to compute the geometric series)
\begin{align}\nonumber
U_n
&= (1-e^{-2\beta})\cdot2\frac{e^{-2\beta}}{(1-e^{-2\beta})^2}
\\
&=\frac{2}{e^{2\beta}-1},\label{eq_un}
\end{align}
which is equivalent to
\[
\beta=\frac12 \log\Big(\frac{2}{U_n}+1\Big).
\]
We can also simplify the entropy:
\begin{align*}
S
&=
-\sum_{l,n}\frac{1}{Z}g(l)\exp(-\beta E(n,l))\log\Big(\frac{1}{Z}\exp(E(n,l))\Big)
\\
&=\frac{2\beta}{e^{2\beta}-1}-\log(1-e^{-2\beta})
\\
&=\frac{1}{Z}Z_l\sum_n2n\beta\exp(-2n\beta)
\\
&\qquad+\frac{1}{Z}Z_n\sum_l\sqrt{l(l+d-2)} \beta g(l)\exp(-\beta\sqrt{l(l+d-2)})
\\
&\qquad+\log(Z)
\\
&=(1-e^{-2\beta})2\beta\frac{e^{-2\beta}}{(1-e^{-2\beta})^2}
\\
&\qquad+\frac{1}{1}Z_l\sum_l\sqrt{l(l+d-2)} \beta g(l)\exp(-\beta\sqrt{l(l+d-2)})
\\
&\qquad+\log(Z_n)+\log(Z_l)
\\
&=\underbrace{2\beta\frac{1}{e^{2\beta}-1}+\log(Z_n)}_{:=S_n}
\\
&\qquad+\underbrace{\frac{1}{Z_l}\sum_l\sqrt{l(l+d-2)} \beta g(l)\exp(-\beta\sqrt{l(l+d-2)})+\log(Z_l)}_{:=S_l}
.
\end{align*}
Now we can express $S_n$ with respect to the cost $U_n$:
\begin{equation}\label{eq_sn}
S_n=\log\Big(1+\frac{U_n}{2}\Big)+\frac{U_n}{2}\log\Big(1+\frac{2}{U_n}\Big).
\end{equation}
At this point the authors use Sterlings formula and the method of steepest descent to calculate the integral version of the discrete sum to obtain
\[
Z_l\approx2\beta^{-d+1}.
\]
The error of this approximation is estimated by numerical results, which indicate that
\[
|\log(Z_l)-\log(2\beta^{-d+1})|=\mathcal{O}\Big(\frac{\eta}{d}\Big).
\]
Using this, the authors obtain
\begin{equation}\label{eq_num1}
U_l=\frac{d-1}{\beta}+\mathcal{O}\Big(\frac{1}{d}\Big)
\end{equation}
and
\begin{equation}\label{eq_num2}
S_l=(d-1)\Big(\log\frac{U_l}{d-1}+1\Big)+\mathcal{O}\Big(\frac{\beta}{d}\Big).	
\end{equation}

A Taylor approximation at $\beta=0$ of \eqref{eq_un} gives us
\begin{equation}\label{eq_Un}
U_n=\frac{1}{\beta}-1+\mathcal{O}(\beta),
\end{equation}
and keeping only the dominant terms in \eqref{eq_sn} results in
\[
S_n=\log\Big(1+\frac{U_n}{2}\Big)+\mathcal{O}(1).
\]
Using the previous results and $U_n\approx 1/\beta$, we get
\[
\tilde{C}=\h\sqrt{\frac{\kappa}{m}}(dU_n-(d-1))+\mathcal{O}(1/d)).
\]
The approximation $U_n\approx 1/\beta$ is mathematically not correct but using \eqref{eq_Un} instead would result in a stricter inequality in the end. So this does not lead to a wrong statement. We can also calculate for the entropy 
\[
S=d\log(U_n)+\mathcal{O}(1),
\]
where the author use $d\log(U_n)+(d-1)\approx d\log(U_n)$
\[
\tilde{C}_\kappa=d\hbar\sqrt{\frac{\kappa}{m}}\left(\exp(\S/d)-\frac{d-1}{d}+\mathcal{O}(1/d^2)\right),
\]
which yields \eqref{eq_sum_cost} in the limit of $d$ to infinity, and finishes the sketch of the proof.

\subsection{Comment on the Proof}\label{ss_comment}
First there are a few physical issues:
Hamiltonians in higher dimension usually model systems with a big number of particles.
But the type of Hamiltonian we allow does not account for interaction between the particles.
This is a particularly big shortcoming, since we also treat problems with small spacial scales, where interaction might be crucial.
Furthermore, the assumption $d\gg1$ makes it useless for many important applications with one or two digit dimension.

The biggest problem is that the statement was not proven by mathematical standards:  the conjecture about the minimality of the Hamiltonian \eqref{eq_H_opt} is unproven and, as mentioned, the authors made assumptions which are based on numerical results \eqref{eq_num1} and \eqref{eq_num2}.

We found two examples that illustrate the limitations of the the energy-surface bound \eqref{dn_bound}:

Our first example (section \ref{sec_ex_1}) violates the inequality for sufficiently large $d$. This contradicts the statement since the violation holds in the limit for $d\to\infty$ and the entropy grows arbitrarily.
Because of this, there must be an error in the argumentation of the paper:
either $H_{opt}$ in \eqref{eq_H_opt} is not the unique optimum in the optimization problem \eqref{eq_min_cost} or $H_{opt}$ itself does not obey the energy-space inequality \eqref{dn_bound}.
The second possibility could be the result of the mathematical imprecisions, we mentioned earlier.
Either way, we found an error at a crucial point in the paper.

The second example (theorem \ref{thm_ex_2}) violates the inequality for small quantum numbers, dimension $1$ and arbitrarily small entropy.
Because of the small entropy, this example might not be important for applications, but it violates the inequality for any factor.
It shows that the energy-surface bound can not be applied at all in this setting, and the limit $d\to\infty$ is crucial.

In the following section, we present the two counterexamples.

\subsection{The Harmonic Oscillator and the Inequality} \label{sec_ex_1}
We will calculate the quantities in the energy surface inequality for a case of the normalized $d$-dimensional harmonic oscillator.
The definitions and their properties can be found in \cite{griffiths05}, section 2.3.
This will give us our first violation of the energy-surface bound and demonstrates how the quantities scale.
The Hamiltonian is 
\begin{align} \label{eq_harm_osc}
H&=\frac{1}{2m}P^2+\frac{m\omega^2}{2} Q^2 -\frac{\hbar\omega d}{2}\1
\\
&=\sum_{k=1}^d \left( \frac{1}{2m} P_k^2+\frac{m\omega^2}{2} Q_k^2 \right) -\frac{\hbar\omega d}{2}\1, \nonumber
\end{align}
which is the sum of one-dimensional Hamiltonians acting on each dimension.
We denote the quantum numbers $n_1,\dotsc,n_d\in\N^0$ as $\vec{n}=(n_1,\dotsc,n_d)$ and get the eigenfunctions 
\begin{align*}
\psid&=\ket{\psi_{n_1}}\dotsb\ket{\psi_{n_d}},
\end{align*}
which are separable, i.e. $\ket{\psi_{n_k}}$ is a function of $x_k$ for all $k=1,\dotsc,d$.
They have eigenenergies
\[
E_{\vec{n}}=\hbar\omega\sum_{k=0}^d(n_k).
\]
Now we calculate the spacial costs of an energy eigenstate:

\begin{lemma}\label{lem_q^2}
We have for all $\vec{n}\in(\N^0)^d$ that
\begin{equation*}
\Braket{\psi_{\vec{n}}|Q^2|\psi_{\vec{n}}}=\frac{\h}{2m\omega}\sum_{k=1}^d(2n_k+1).
\end{equation*}
\end{lemma}

\begin{proof}
The definitions and their unproven properties are taken from (\cite{griffiths05}, part I, 2.3).
We define creation and annihilation operators $a_k^\dag$, $a_k$, $k\in\{1,\dotsc,d\}$, which act on $\mathcal{L}_2(\R^d)$ as
\begin{align*}
a_k^\dag&= \sqrt{\frac{1}{2\h m\omega}}\big(m\omega Q_k-iP_k\big)
\\
a_k&= \sqrt{\frac{1}{2\h m\omega}}\big(m\omega Q_k+iP_k\big).
\end{align*}
They have the properties
\[
a_k^\dag\ket{\psi_{n_1,\dotsc,n_d}} =\sqrt{n_k+1} \ket{\psi_{n_1,\dotsc,n_k+1,\dotsc,n_d}}
\]
 and
\[
a_k\ket{\psi_{n_1,\dotsc,n_d}} =\sqrt{n_k} \ket{\psi_{n_1,\dotsc,n_k-1,\dotsc,n_d}}.
\]
For the commutator we have $[a_k,a_k^\dag]=\1$.
One easily sees that
\begin{equation}
Q_k=\sqrt{\tfrac{\h}{2m\omega}}(a_k^\dag+a_k) \label{x_constr}
\end{equation}

With \eqref{x_constr} the claim follows directly:
\begin{align*}
\Braket{\psi_{\vec{n}}|Q^2|\psi_{\vec{n}}}
&=\sum_{k=1}^d\Braket{\psi_{\vec{n}}|Q_k^2|\psi_{\vec{n}}}
\\
&=\sum_{k=1}^d\Braket{\psi_{{n_k}}|Q_k^2|\psi_{{n_k}}}
\\
&=\frac{\h}{2m\omega}\sum_{k=1}^d\Braket{\psi_{{n_k}}|(a_k^\dag+a_k)^2|\psi_{{n_k}}}
\\
&=\frac{\h}{2m\omega}\sum_{k=1}^d\Braket{\psi_{{n_k}}|a_k^\dag a_k+a_ka_k^\dag|\psi_{{n_k}}}
\\
&=\frac{\h}{2m\omega}\sum_{k=1}^d (2n_k+1).
\end{align*}
\end{proof}
As example state we choose the equal distribution over all eigenstates that have quantum numbers $\leq l$:
\begin{align}\label{eq_ex_1}
\rho=\frac{1}{\alpha}\sum_{\vec{n}\in A}\ketbra{\psi_{\vec{n}}},
\end{align}
with $A=\lbrace\vec{n}|n_1,\dotsc,n_d\in \{0,\dotsc,l\}\rbrace$, $\alpha=|A|=(l+1)^d$.
By rearranging the summation order we calculate its energy and spacial costs:
\begin{align}
\tr(\rho H) \nonumber
&=\frac{1}{\alpha}\sum_{\vec{n}\in A}\Braket{\psi_{\vec{n}}|H|\psi_{\vec{n}}}
\\\nonumber
&=\frac{1}{\alpha}\sum_{\vec{n}\in A}\sum_{k=0}^d\hbar\omega n_k
\\\nonumber
&=\frac{d}{l+1}\sum_{k=0}^l \hbar \omega k
\\\label{eq_cost_en}
&=\frac{\hbar\omega}{2} d l
\end{align}
and with lemma \ref{lem_q^2} we get
\begin{align}
\tr(\rho Q^2) \nonumber
&=\frac{1}{\alpha}\sum_{\vec{n}\in A}\Braket{\psi_{\vec{n}}|Q^2|\psi_{\vec{n}}}
\\\nonumber
&=\frac{\hbar}{2m\omega\alpha} \sum_{\vec{n}\in A}\sum_{k=0}^d (2n_k+1)
\\\nonumber
&=\frac{\hbar}{2m\omega}\Big(d+1+ \frac{2}{\alpha}\sum_{\vec{n}\in A}\sum_{k=0}^d (n_k)\Big)
\\\nonumber
&=\frac{\hbar}{2m\omega} \Big(d+1+2\frac{d}{l+1}\frac{(l+1)l}{2}\Big)
\\\label{eq_cost_sp}
&=\frac{\hbar}{2m\omega}(dl+d+1).
\end{align}

Now for the right hand side of the inequality \eqref{dn_bound}. We calculate the entropy
\begin{align*}
\S(\rho)&=-\sum_{\vec{n}\in A}\frac{1}{\alpha}\log\Big(\frac{1}{\alpha}\Big)\\
&=\log(\alpha)\\
&=d\log(l+1),
\end{align*}
and get
\begin{align*}
\frac{\hbar^2}{2m} d^2\left(e^{\S(\rho)/d}-1\right)^2&=\frac{\hbar^2}{2m}d^2l^2,
\end{align*}
so the inequality \eqref{dn_bound} reads
\[
\frac{\hbar^2}{4m}dl(dl+d+1)\geq\frac{\hbar^2}{2m}d^2l^2,
\]
which is equivalent to
\begin{equation}\label{eq_ex_1_res}
\frac{1}{2}(dl+l+1)\geq dl.
\end{equation}

We clearly have a violation for $d\to\infty$. So the inequality can be violated for arbitrary big entropy and big enough dimension. This violation can become arbitrarily close to a factor 2.
But ignoring this error, we see that both sides in this example of the inequality have the same scaling in the limit $d$ to infinity.
With a better choice of $\rho$, one can obviously achieve a larger violation, but we couldn't find better computable examples in the context of this work.

\subsection{Further Counterexample}

We have already violated the inequality by a factor 2. By optimizing the state $\rho$ this factor can easily be increased. Next we even show that the inequality is asymptotically wrong for small dimensions $d$, meaning that it will hold for no constant factor.

\begin{theorem}[Counterexample for arbitrary factor]\label{thm_ex_2}
For all $C>0$, the bound
\[
\tr[\rho H] \, \tr[\rho Q^2] \geq C \frac{\h^2}{2m}(\exp(\S(\rho))-1)^2 
\]
can be violated by the one-dimensional harmonic oscillator. The entropy of $\rho$, which is diagonal in the eigenbasis of $H$, can become arbitrarily small.
\end{theorem}

\begin{proof}
We write the inequality \eqref{dn_bound} for the one-dimensional harmonic oscillator and the pa\-ram\-e\-ter\-ized state
\[
\rho(p)=(1-p)\ketbra{\phi_0}+p \ketbra{\phi_1},\qquad p\in [0,1].
\]
For convenience we define $f$ and $g$ as the right hand side and the left hand side with respect to $\rho(p)$, respectively:
\begin{align*}
f(p)
&=\frac{\h^2}{2m}(\exp(\S(\rho(p)))-1)^2
\\
&=\frac{\h^2}{2m}\left( e^{-p\log(p)-(1-p)\log(1-p)}-1\right)^2,
\\
g(p)
&=\tr\left(\rho(p) H\right) \tr\left(\rho(p) Q^2\right)
\\
&=\frac{\h^2}{2m}p((1-p)+3p)\\
&=\frac{\h^2}{2m}(2p^2+p).
\end{align*}
Now we will show 
\[
\lim_{p\to0^+}\frac{f(p)}{g(p)} = \infty,
\]
which proves our statement.\\
Using the Landau notation for $p\to 0^+$, we have
\begin{align*}
f(p)&=\frac{\h^2}{2m}\left(\ksumi \frac{(-p\log(p))^k}{k!}
\cdot\ksumi \frac{\left( (p\log(1-p)\right)^k}{k!} 
\cdot\ksumi p^k-1 \right)^2\\
&=\frac{\h^2}{2m}\Big(\left(1-p\log(p)+\tfrac12 p^2\log(p)^2+o(p^2)\right)\\
&\quad\cdot\left( 1-p\log(1-p)+o(p^2)\right)\cdot\left(1+p+p^2+o(p^3)\right)-1\Big)^2\\
&=\frac{\h^2}{2m}\left(p+p^2+p\log(1-p)-p\log(p)(1+p)+o(p^2)\right)^2\\
&=\underbrace{\frac{\h^2}{2m}p^2\left(1-2\log(p)+\log(p)^2\right)}_{:=\tf(p)}+o(p^2),
\end{align*}
where we used that $p\log(p)^k, \log\left(1-p\right), \log(p)\log(1-p)\in o(1)$ for all \mbox{$k\in\N$}.
Now we have $\lim_{p\to0^+}\frac{f(p)}{g(p)}=\lim_{p\to0^+}\frac{\tf(p)}{g(p)}$ if one of the limits converges.\\
To calculate the limit, we apply L'Hôpital's rule. All requirements are fulfilled: on $\R^+$ we have that $\tf,g$ are differentiable with $\tf'(p)=2p(-\log(p)+\log(p)^2)$ and $g$ is non-zero. We also check that $\limp\tf(p)=\limp g(p)=0$ and 
\[
\limp\frac{\tf'(p)}{g'(p)}=\limp\frac{2(-\log(p)+\log(p)^2)}{4+\frac{1}{p}}=+\infty.
\]
So by L-Hôpital's rule 
\[
\limp \frac{f(p)}{g(p)}=\limp \frac{f'(p)}{g'(p)}=+\infty,
\]
which concludes the proof.
\end{proof}
\section{Variational Principle for Gibbs States}
\label{sec_GVP}
In this chapter we prove a useful theorem about states with optimal entropy under linear constraints.
It was used in its infinite-dimensional version in section \ref{ss_justific}, eqation \eqref{eq_proof_gibbs} and plays an important role in section \ref{sec_findimap}.
This statement for finite as well as for infinite dimensions but only for one constraint is proven in \cite{carlen10}, theorem 1.3.
There the author used a different method to prove it.

The proof in this section is based on the arguments in \cite{neumann27}, p.279 ff.

\begin{remark*}
Let $\mathcal{H}$ be a finite dimensional Hilbert space. In the following we will identify $\mathcal{H}$ with $\C^n$
\begin{itemize}
	\item The set of self-adjoint operators on $\mathcal{H}$ is denoted by $\mathbf{H}=\{\rho\in\Cn|\rho^\dag=\rho\}$.
$\mathbf{H}$ is a $n^2$-dimensional $\R$-subspace of $\mathcal{B(H)}$.
	\item We set $[a,b]_\mathcal{H}:=\{\rho\in\mathcal{B(H)}|a\1\leq\rho\leq b\1\}$.
	\item $\Cn$ equipped with the Hilbert–Schmidt inner product
\[
\langle A,B\rangle = \tr \left(A^\dag B\right)
\]
for $A,B\in \Cn$ is a Hilbert space.
	\item If $M\subset \Cn$ is an affine space, then $\inter_M(X)$ is the interior of $X\subseteq M$ with respect to the subspace topology.
	$\Cn$ is equipped with the topology induced by the trace norm.
	 For $X\subseteq\Cn$ we set $\inter(X):=\inter_A(X)$, with A being the smallest affine space containing $X$.
	 Analogously we set $\dim(X):=\dim(A).$
	\item The operator norm on $\Cn$ will be denoted as $\VERT{\cdot}$.
	\item For a differentiable function $f:\Cn\to\C$ and a matrix $\rho=(\rho_{ij})_{ij}\in\Cn$ we define the gradient of $f$ in $\rho$ by
\[
\langle\nabla f(\rho),\sigma \rangle=\mathrm{d}f(\rho)\sigma
\]
for all $\sigma\in\Cn$.
We used the total differential ${\rm d}f({\rho})\colon\C^n\to\C,\sigma\mapsto\partial_{\sigma}f({\rho})=\left.\frac{d}{dt}f(\rho+t\sigma)\right|_{t=0}$.
\end{itemize}

\end{remark*}
\begin{lemma}\label{lem_diff}
Let $\mathcal{H}$ be a finite-dimensional Hilbert space.
For an analytic function $f:[a,b]\to\C$ with $\rho\in\inter_\mathbf{H}([a,b]_\mathcal{H})$ and $V\in\Cn$, $\VERT{V}\leq1$ we have
\[
\left.\frac{d}{dt} \tr\left(f(\rho+tV)\right)\right\vert_{t=0}=\tr\left(Vf'(\rho)\right).
\]
In the appendix (lemma \ref{lem_matr_ps}) we show that $f$ and $f'$ are well defined in this equation.
We assume that $f$ has a series expansion around $x_0\in\R$ with convergence radius $r>0$ such that $(a,b)\subseteq B_{r}(x_0)$ and that f converges absolutely on the convergence radius.

\end{lemma}

\begin{proof}
We first show the statement for monomials. By linearity it is then also true for polynomials. We will then show that it extends to all analytic functions by uniform convergence.

Now let $f:\rho\mapsto\rho^k$ for a fixed $k\in\N$ and $V\in\Cn$, then
\[
\ddt\, \left. \tr\left(f(\rho+tV)\right)\right|_{t=0}=\tr\left(\rho^{k-1}V\right)+\tr\left(\rho^{k-2}V\rho\right)+\dots+\tr\left(V\rho^{k-1}\right)
\]
since all other summands are zero due to $t=0$. But by the cyclicity of the trace the remaining summands are identical and we get
\[
\ddt\, \left.\tr\left(\left(\rho+tV\right)^k\right)\right|_{t=0}=k\tr\left(V\rho^{k-1}\right).
\]
This proves our statement for all monomials $f$ and thus for all polynomials.
Now we generalize this to $t_0\in\R$:
\begin{align*}
\frac{d}{dt}\left. \tr\left(f(\rho+tV)\right)\right|_{t=t_0}&=\frac{d}{dt}\left. \tr\left(f(\rho+t_0V+tV)\right)\right|_{t=0}\\
&=k\tr\left(V\left(\rho+t_0V\right)^{k-1}\right),
\end{align*}

Now let $f$ be an analytic function with coefficients $(a_k)$ and convergence radius $r>0$ around 0 and $(a,b)\subseteq B_{r}(x_0)$:
\[
f(x)=\sum_{k=0}^\infty a_kx^k \qquad \forall |x| <r.
\]
This covers the general case, since all functions $g(x)=\sum_{k=0}^\infty b_k(x-x_0)^k$ with convergence radius $r$ around $x_0$ can be shifted to $\tilde{g}(x):=g(x+x_0)=\sum_{k=0}^\infty b_kx^k$, which has convergence radius $r$ around $0$.
Now we have for all $X\in[x_0-r,x_0+r]_\mathcal{H}$ that $g(X)=\tilde{g}(X-x_0\1)$ with $(X-x_0\1)\in[-r,r]_\mathcal{H}$. Analogously we have $g'(X)=\tilde{g}'(X-x_0\1)$, so if we have shown the desired statement for $\tilde{g}$ it translates to $g$.

Now back to $f$: $f'$ also has convergence radius $r$ around $0$ and is given by
\begin{equation}
f'(x)=\sum_{k=1}^\infty k a_k x^{k-1} \qquad \forall |x|<r. \label{eq_f_prime}
\end{equation}

For $\rho\in[a,b]_\mathcal{H}$, $I:=[-\eps,\eps]	$, $\eps=\frac12 (r-\VERT{\rho})$, $n\in\N$ we define the partial sums $g_n:I\to\C$ by
\[
g_n(t)=\tr( \sum_{k=0}^n a_k (\rho+tV)^k).
\]
Note that $\epsilon>0$ since $\rho\in\inter_\mathbf{H}([a,b]_\mathcal{H})$.
We have already proven that 
\[
g_n'(t)=\tr(V\sum_{k=1}^n ka_k(\rho+tV)^{k-1}).
\]
Let $\|\cdot\|_I$ denote the uniform norm on $I$. At the end we will show that the sum of the uniform norm of the summands of $g_n'$ is finite, i.e.
\begin{equation}
\sum_{k=1}^\infty\left\| k a_k \tr(V(\rho+tV)^{k-1})\right\|_I<\infty. \label{eq_sum_norms}
\end{equation}
By the Weierstrass M-test we then have that $g_n'$ converges uniformly on $I$ to a function $h$.
But we also have $\lim_{n\to\infty}g_n(t)=\tr(f(\rho+tV))$ for all $t\in I$ pointwise:
by definition $\VERT{\rho+tV}<r$, so both sides of the equation converge and we name the limit $g(t)$.

To sum things up we then have on $I$ that $g_n'\to h$ uniformly and $g_n\to g$ pointwise.
By the theory of uniform convergence (\cite{howie01}, theorem 7.11) we then have that $g$ is differentiable and has derivative $g'=h$, which is our statement.
The only thing left to show is \eqref{eq_sum_norms}, which we will do now.
For the summands we have:
\begin{align*}
\left\|k\,a_k\tr(V(\rho+tV)^{k-1})\right\|_I
&\leq
k\, d\, \lvert a_k\,\rvert\, \VERT{\rho+tV}^{k-1}
\\
&< k\, d\, \lvert a_k\rvert\, \left(\VERT{\rho}+\eps\right)^{k-1}
\\
&= k\, d\, \lvert a_k\rvert\, \left(r-\eps\right)^{k-1}
\end{align*}
So we can use that $f'$ converges within $r$ (c.f. \eqref{eq_f_prime}):
\[
\sum_{k=1}^\infty\left\| k a_k \tr(V(\rho+tV)^{k-1})\right\|_I
<d\sum_{k=1}^\infty k\,
\lvert a_k\rvert\, \left(r-\eps\right)^{k-1}<\infty,
\]
which completes the proof.
\end{proof}

\begin{lemma}[Gradient of von Neumann entropy]\label{lem_gr_vNe}
Let $\mathcal{H}$ be a finite-dimensional Hilbert space.
We have that the von Neumann entropy $\S$ is well defined on $[0,1]_\mathcal{H}$ and particularly on $\mathcal{S(H)}$. For all $\rho\in\inter_\mathbf{H}([0,1]_\mathcal{H})$ we have
\[
\nabla\S(\rho)= (-\1  -\log(\rho))^\dag.
\]
\end{lemma}
\begin{proof}
We set $f:[0,1]\to\R,x\mapsto -x\log(x)$.
We show in  \ref{vnEntropy} that $f$ is well defined on $[0,1]_\mathcal{H}$ and we have $\S(\rho)=\tr(f(\rho))$.
The derivative is $f':[0,1]\to\R, x\mapsto (-1-\log(x))$ and by lemma \ref{lem_diff}, we have for all $V\in\Cn$, $\VERT{V}\leq1$, $\rho\in\inter_\mathbf{H}([0,1]_\mathcal{H})$
\begin{align*}
\langle\nabla \S(\rho), V\rangle&
=
\frac{d}{dt}\left.\S(\rho+tV)\right|_{t=0}
\\
&=\frac{d}{dt}\left.\tr\left(f(\rho+tV)\right)\right|_{t=0}
\\
&=\tr\left(Vf'(\rho)\right)
\\
&=\langle f'(\rho)^\dag ,V\rangle
\\
&=\langle (-\1-\log(\rho))^\dag ,V\rangle.
\end{align*}
This gives us the desired statement.
\end{proof}

\begin{theorem}[Variational principle for Gibbs states]\label{thm_gibbs}
Let $\mathcal{H}$ be a finite-dimensional Hilbert space.
Let $H_1,H_2\in\Cn$, $H_1,H_2\geq0$, $C_1,C_2 \in \R$. A local optimum of $\S(\rho)$ over  $\rho\in\mathcal{S(H)}$ under the constraints 
\begin{align*}
g_1(\rho):=\tr(\rho H_1)&=C_1\\
g_2(\rho):=\tr(\rho H_2)&=C_2\\
\end{align*}
can be written as a Gibbs state
\begin{equation} \label{eq_gibbs_state}
\rhob:=\frac{\exp(-\beta_1 H_1-\beta_2 H_2)}{\tr(\exp(-\beta_1 H_1-\beta_2 H_2))}
\end{equation}
for $\beta_1,\beta_2\in\R$.
\end{theorem}

\begin{proof}
The the proof is based on the one in \cite{neumann27}, p. 279 ff. We will use that $\S$ is strictly concave, which is proven in (\cite{carlen10}, theorem 2.10).

We optimize over $[0,1]_\mathcal{H}$ and add a third constraint $g_3(\rho):=\tr(\rho)=C_3$. This is more general since our statement then results for $C_3=1$.

First suppose that $\rho\in\inter_\mathbf{H}([0,1]_\mathcal{H})$. This is a non-empty, convex set.
Let $\rho$ be locally optimal.
By the method of Lagrange multipliers we have
\[
\nabla\S(\rho)\in \spa\{\nabla g_1,\nabla g_2,\nabla g_3\}.
\]
This means there are $\lambda_1,\lambda_2,\lambda_3\in\C$ s.th.
\[
\nabla\S(\rho)=\lambda_1 H_1+\lambda_2 H_2+ \lambda_3\1.
\]
But since all matrices are from the $\R$-vector space of self-adjoint matrices, we can choose the coefficients $\lambda_1,\lambda_2,\lambda_3$ real.
So we get
\[
\log(\rho)=-\lambda_1 H_1-\lambda_2 H_2- (1+\lambda_3)\1
\]
and
\begin{align*}
\rho&=\exp\left(-\lambda_1 H_1-\lambda_2 H_2- (1+\lambda_3)\1 \right)\\
&=\exp(-1-\lambda_3) \exp\left( -\lambda_1 H_1+\lambda_2 H_2\right).
\end{align*}
For $C_3=1$ we have $\exp(-1-\lambda_3)=\tr({\exp\left( -\lambda_1 H_1+\lambda_2 H_2\right)})^{-1}$,
and get the statement after renaming the parameters $\lambda_1$ and $\lambda_2$.

Now suppose $\rho\in\mathcal{S(H)}\backslash\inter_\mathbf{H}([0,1]_\mathcal{H})$.
We denote the set of feasible points by
\[
M=\{\sigma\in[0,1]_\mathcal{H}|g_i(\sigma)=C_{i}\}.
\]
We know that $X:=M\cap\inter_\mathbf{H}([0,1]_\mathcal{H})=\emptyset$, because otherwise the maximum of the concave function $\S$ over the convex set $M$ would be in $X$.
Since $\dim([0,1]_\mathcal{H})=n^2$ and $g_i^{-1}(C_i)$ defines an affine hyperplane for $i=1,2,3$, we have that for all $\eps>0$ there are $C_{i,\epsilon}$,  $i=1,2,3$ such that 
\[
X_\epsilon:=\underbrace{\{\sigma\in[0,1]_\mathcal{H}|g_i(\sigma)=C_{i,\epsilon}\}}
_{=:M_\epsilon}
\cap\inter_\mathbf{H}([0,1]_\mathcal{H})\neq\emptyset.
\]
As we have seen, the maximum on $M_\epsilon$ is in $X_\epsilon$, thus it can be written as a Gibbs state 
\[
\rho_\epsilon=C_{3,\epsilon}\exp(-\beta_{1,\epsilon} H_1-\beta_{2,\epsilon} H_2)/\tr(\exp(-\beta_{1,\epsilon}H_1-\beta_{2,\epsilon}H_2)).
\]
Now we know that a subsequence of a sequence $(\rho_{1/n})_n$ converges to a state $\tilde{\rho}\in M$, since $M=\lim_{n\to\infty}M_{1/n}$ and $[0,1]_\mathcal{H}$ is a closed set.
After renaming the sequence $(\rho_{1/n})_n$ we can assume that it converges itself to $\tilde{\rho}$.
Thus $\beta_{1,\epsilon}$ and $\beta_{2,\epsilon}$ converge to some points $\beta_{1},\beta_{1}$, respectively and we have
\[
\tilde{\rho}=C_{3}\exp(-\beta_{1} H_1-\beta_{2} H_2)/\tr(\exp(-\beta_{1}H_1-\beta_{2}H_2)).
\]
Because $\S$ is continuous, $\tilde{\rho}$ is an optimum $M$ and since $\S$ is strictly concave we have $\tilde{\rho}=\rho$.
For $C_3=1$ this shows our statement in this case.

Now we have shown our statement on all of $[0,1]_\mathcal{H}$ and thus on $\mathcal{S(H)}$.
\end{proof}
\section{Finite Dimensional Approximations} \label{sec_findimap}
Instead of studying $\mathcal{L}_2(\R)$, which is a infinite-dimensional Hilbert space, one can consider the finite dimensional case $\Cn$, which is numerically easier to handle.
Sometimes $\mathcal{L}_2(\R)$ can be modelled as the  limit of $\Cn$ for $n \to\infty$ to deduce results from the finite dimensional space.
This procedure is often referred to as the thermodynamic limit.

In this section we will first introduce the finite dimensional analogous of the bound discussed in section \ref{sec_bound}.
Then we present numerical results that result from it.
For completeness, the source-code that generated the data is in appendix \ref{app_code}.

\subsection{Preliminaries}
We search an alternative energy-space bound $f:\R\to\R$ such that
\begin{equation}\label{eq_alt_bound}
\tr(\rho Q^2)\tr(\rho H)\geq f(\S(\rho))
\end{equation}
for the finite-dimensional case (i.e. $\dim(\mathcal{H})=d<\infty$) of the problem described in the preliminaries \ref{sec_prel}.
This $d$ is not to be confused with its use in section \ref{sec_bound}, where it described the number of degrees of freedom.
We can assume without loss of generality that $\mathcal{H}=\C^{d}$ and adjust the other definitions accordingly:
A Hamiltonian $H\in\C^{d\times d}$ is self-adjoint and can be normalized to $H\geq0$.
A matrix $\rho\in\C^{d\times d}$ is called state if it has $\rho\geq0$ and $\tr\rho=1$.
The setting for the finite dimensional case is taken from \cite{torre2003}, which also contains more explanations on the definitions.
For convenience we define the $d$-th root of unity as $\omega=\exp(i\frac{2\pi}{d})$.
Our observables in this case are the position operator
\[
Q=\sum_{k=-\tfrac{n-1}{2}}^{\tfrac{n-1}{2}}k\ketbra{\phi_k},
\]
and the momentum operator
\[
P=\sum_{k=-\tfrac{n-1}{2}}^{\tfrac{n-1}{2}}
\omega^k\ketbra{\psi_k},
\]
in their respective spectral decomposition.
\\
They have the characteristic properties
\[
B=\exp\Big(i\frac{2\pi}{d}Q\Big),
\]
and
\[
T=\exp(-iP),
\]
where 
\[
T=\sum_{k=-\tfrac{n-1}{2}}^{\tfrac{n-1}{2}-1}
|\phi_{k+1}\rangle\langle{\phi_{k}}|
+(-1)^{d+1}\Big|\phi_{-\tfrac{n-1}{2}}\Big\rangle\Big\langle{\phi_{\tfrac{n-1}{2}-1}}\Big|
\]
is the translation operator and
\begin{align*}
B&=\sum_{k=-\tfrac{n-1}{2}}^{\tfrac{n-1}{2}-1}
|\psi_{k+1}\rangle\langle{\psi_{k}}|
+(-1)^{d+1}\Big|\psi_{-\tfrac{n-1}{2}}\Big\rangle\Big\langle{\psi_{\tfrac{n-1}{2}-1}}\Big|,
\end{align*}
is the momentum boost.
\\
We assume that the Hamiltonian can be written as $H=\frac{1}{2} P^2+V(Q)$, where $V$ is given as a Laurent series:
\[
V(Q)=\sum_{k=-\infty}^\infty a_k Q^k
\]
for $a_k\in\R$.

The parity operator is defined as 
\[
\Pi = \sum_{k=-\tfrac{n-1}{2}}^{\tfrac{n-1}{2}}\ketbras{\phi_{k}}{\phi_{-k}}.
\]
Thus we have $\Pi^\dag=\Pi=\Pi^{-1}$.

Although the author can't show that $\tr(\rho Q)=0$ as we did in the infinite-dimensional case, we show that $\rho$ can be chosen symmetric if $H$ is symmetric:
\begin{lemma}
We suppose that $[H,\Pi]=0$. Then the state $\rho$ which maximizes $\S(\rho)$ under the constraints $\tr(\rho H)=C_1$, $\tr(\rho Q^2)=C_2$ is symmetric, i.e.
\[
[\Pi,\rho]=0.
\]
\end{lemma}
\begin{proof}
Let $\rho$ be a maximizing state. There exists one, since the set of all states is closed. Now we define a state $\hat{\rho}:=\frac12(\Pi\rho\Pi+\rho)$, which is symmetric
\[
\Pi\hat{\rho}\Pi=\hat{\rho}
\]
and hast the same costs as $\rho$:
\[
\tr(\hat{\rho}H)=\frac12(\tr(\Pi\rho\Pi H)+\tr(\rho H))=\tr(\rho H),
\]

\[
\tr(\hat{\rho}Q^2)=\frac12(\tr(\Pi\rho\Pi Q^2)+\tr(\rho Q^2))=\tr(\rho Q^2).
\]
Thus $\hat{\rho}$ fulfils the same constraints as $\rho$.
We have that 
\[
\S(\tilde{\rho})\geq \S(\rho)
\]
because $\S$ is a concave function (\cite{carlen10}, theorem 2.10) and $\S(\rho)=\S(\Pi\rho\Pi)$.
Thus $\tilde{\rho}$ is also a maximizing state.
\end{proof}

\subsection{Numerical Results}\label{ss_numres}
To find an alternative energy-space bound \eqref{eq_alt_bound} we 
approximate 
\begin{equation}\label{eq_finite_min}
\max_{\rho,H}\S(\rho)
\end{equation}
over $\rho\in\mathcal{S(H)}$ and all Hamiltonians $H=P^2+V(Q)$ under the constraint $\tr(\rho H)\tr(\rho Q^2)=C$.
In order to make it computable we relaxed the condition that $\rho$ is diagonal in the eigenstates of $H$.
This enables us to apply the Gibb's variational principle (\ref{thm_gibbs}) and reduces the optimization parameters from $d$ to 2: instead of optimizing over all $(\lambda)_{i=1}^d$ in $\rho=\sum_{i=1}^d\lambda_i\ketbra{\psi_i}$, we optimize over $\beta_1$ and $\beta_2$ of $\rho_{\beta_1,\beta_2}$.
This still solves our problem \eqref{eq_finite_min} because every optimum there can also be written as a Gibbs state \eqref{eq_gibbs_state} with respect to the Hamiltonians $H$ and $Q^2$.
We write for the costs and the entropy,
\begin{align*}
C_H(\beta_1,\beta_2)
&:=\tr(\rho_{\beta_1,\beta_2}H)\tr(\rho_{\beta_1\beta_2}Q^2),
\\
\S(\beta_1,\beta_2)
&:=\S(\rho_{\beta_1,\beta_2}).
\end{align*}

To find a feasible bound, we plot $(\S(\beta_1,\beta_2),C_H(\beta_1,\beta_2 ))$ for a choice of $\beta_1,\beta_2\in\R$ and Hamiltonians $H$.
Now for any suitable $f$, the reflected graph $\{(f(x),x)|x\in\R^+\}$ must govern this set of points.
The parameters $\beta_1,\beta_2$ are chosen such that the plot is representative for costs in the interval  $1\leq C_H(\beta_1\beta_2)\leq 100$.
For $\beta_1$ we used 300 points, equally distributed in $[-5,5]$, while for $\beta_2$ we used 200 points, equally distributed in $[-0.5,2]$.
A dynamical termination criterion is implemented, which avoids the calculation of small/large $\beta_1$ and $\beta_2$ if the costs are outside of our boundaries.
This is based on the assumption that $C_H(\beta_1,\beta_2)$ is monotonically decreasing in both variables.

Considering that the problem is computationally expensive, we limited the choice of Hamiltonians to two types:
Firstly monomials in $Q$:
\[
H=\frac{1}{2}P^2+\sgn(n)\cdot\theta\cdot |Q^n|,
\]
with $\theta\in\{0.1, 0.5,1,5,10\}$ and $n\in\{-3,-2,\dotsc,5\}$.
The dimensions are $d\in\{50,100\}$.
The second type Hamiltonians are given by a Laurent series in $Q$:
\[
H=\frac{1}{2} P^2+\sum_{n=-2}^2a_n\sgn(n)\cdot\theta\cdot |Q^n|,
\]
with $a_n\in \{0.1,1,5\}.$

Figure \ref{fig} shows the results in a scatter plot.
This means, each point represents a state in a certain Hamiltonian.
It also contains the graph of
\begin{equation}\label{eq_nbound}
f^{-1}(x)=\log(\alpha\sqrt{x}+1),
\end{equation}
with $\alpha=2.3455$, which is a good upper bound in all areas of the plot. This $\alpha$ is the smallest coefficient such that the graph of $f^{-1}$ is an upper bound for all points.
This results in
\[
f(x)=\Big(\alpha^{-1}e^x-1\Big)^2
\]
as a new candidate for the energy-surface bound, where $\alpha^{-1}=0.4263$.
This coefficient was optimal for dimension 50 as well as for dimension 100.
For $\alpha=2$, which results in a stronger bound we modify the original bound \eqref{dn_bound} to
\[
\tr\left(\rho H\right)  \tr(\rho Q^2) \geq \frac{\h^2d^2}{2m}(\frac{1}{2}\exp(\S(\rho)/d)-1)^2.
\]
Plugging in the quantities from our first example \eqref{eq_ex_1}, one sees that it already obeys this bound.
In this case the right hand side reduces to 
\[
\frac{\h^2d^2}{8m}(l-1)^2,
\]
which results in a trivial statement for the second example (theorem \ref{thm_ex_2}).
So even if a modified Version with a factor $\alpha$ holds in the limit for $d\to\infty$, it would be very weak in this scale.
\begin{figure}[th!]
    \centering
    \includegraphics[trim=6.5in 1in 5.5in 2in, clip=true, width=\textwidth]{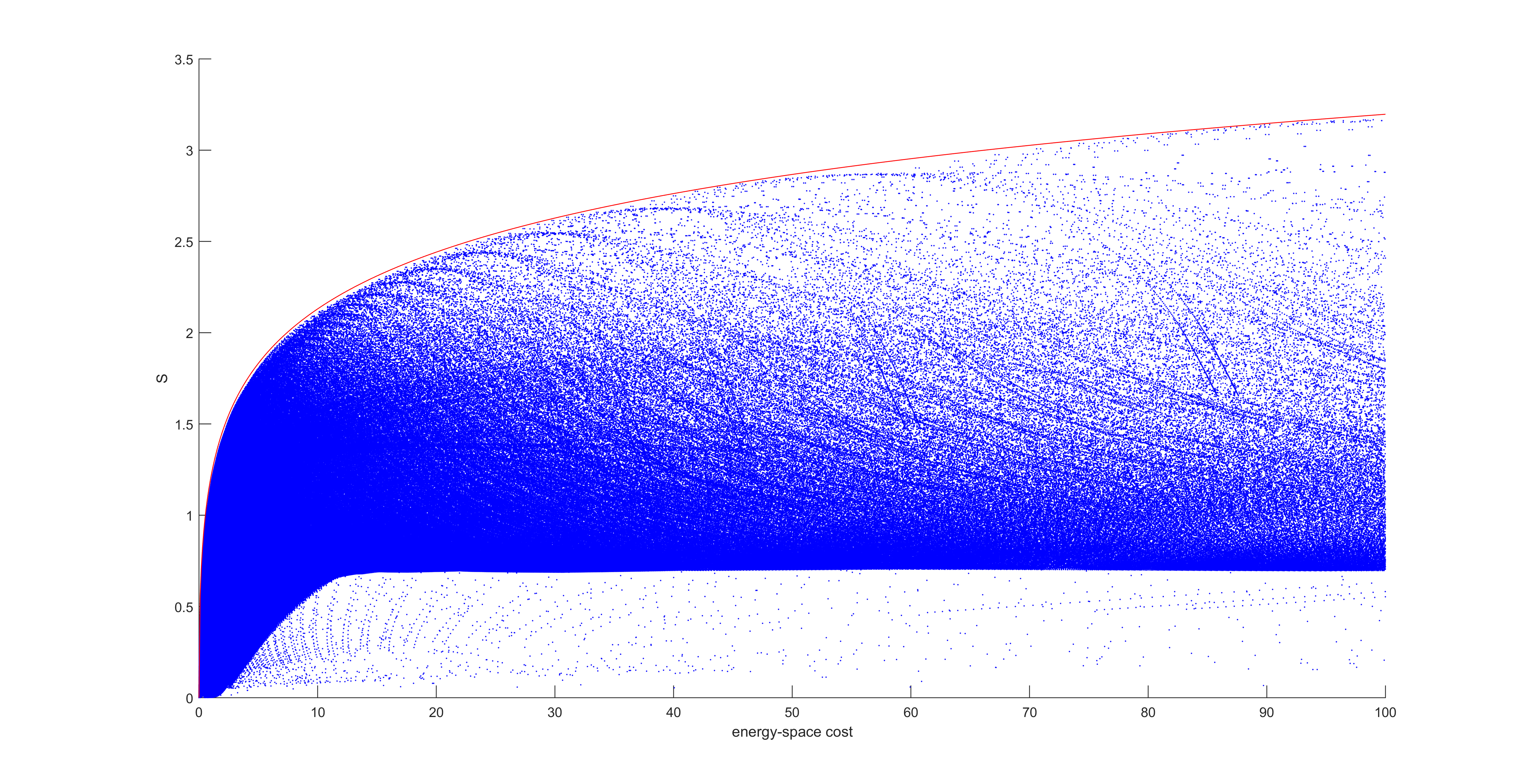}
    \caption{Plot of our selection of states and Hamiltonians (blue dots) as discribed in section \ref{ss_numres}. The red line is a potential bound function \eqref{eq_nbound} with optimal coefficient.}\label{fig}
\end{figure}
\section{Analytic Approach with the Relative Entropy}
In this section we present an analytic approach to find a lower bound for the space-energy costs in finite dimension.
We will also illustrate why it didn't lead to satisfying results.

For two density matrices $\rho,\sigma\in\mathcal{S(H)}$ the quantum relative entropy of $\rho$ with respect to $\sigma$ is defined by
\begin{equation*}
\S(\rho\|\sigma)=-\tr(\rho\log(\sigma))-\S(\rho).
\end{equation*}
We assume by convention $-s\cdot\log(0)=\infty$ for any $s>0$.
This implies $\S(\rho\|\sigma)=\infty$ in case of $\supp(\rho)\cap\ker(\sigma)\neq\{0\}$. One can show that the relative entropy is non-negative at all times. It vanishes if and only if $\rho=\sigma$.

Now we can derive a neat expression for the energy and spacial cost:
\begin{align*}
\S\Big(\rho\|\frac{e^{-H}}{\tr(e^{-H})}\Big)
&=
-\tr\Big(\rho\log\Big(\frac{e^{-H}}{\tr(e^{-H})}\Big)\Big) -\S(\rho)
\\
&=-\tr(\rho\log(e^{-H}))+\tr(\rho\log(\tr(e^{-H})))-\S(\rho)
\\
&=\tr(\rho H)+\log(\tr(e^{-H}))\tr(\rho)-\S(\rho)
\\
&=\tr(\rho H)+\log(\tr(e^{-H}))-\S(\rho),
\end{align*}
thus
\begin{equation*}
\tr(\rho H)=\S(\rho)-\log(\tr(e^{-H}))+\S\Big(\rho\|\frac{e^{-H}}{\tr(e^{-H})}\Big).
\end{equation*}

We apply this to the energy-surface costs:
\begin{align*}
\tr(\rho H)\tr(\rho Q^2)
&=
\Big(\S(\rho)-\log(\tr(e^{-H}))+\S\Big(\rho\|\frac{e^{-H}}{\tr(e^{-H})}\Big)\Big)
\\
&\qquad\cdot\Big(\S(\rho)-\log(\tr(e^{-Q^2}))+\S\Big(\rho\|\frac{e^{-Q^2}}{\tr(e^{-Q^2})}\Big)\Big)
\\
&=\S(\rho)^2-\S(\rho)\big(\log(\tr(e^{-H}))
\\
&\qquad+\log(\tr(e^{-Q^2}))\big)+\log(\tr(e^{-H}))\log(\tr(e^{-Q^2}))
\\
&\qquad+ \S\Big(\rho\|\frac{e^{-H}}{\tr(e^{-H})}\Big)
\big(\S(\rho)-\log(\tr(e^{-Q^2}))\big)
\\
&\qquad+ \S\Big(\rho\|\frac{e^{-Q^2}}{\tr(e^{-Q^2})}\Big)
\big(\S(\rho)-\log(\tr(e^{-H}))\big)
\\
&\qquad+ \S\Big(\rho\|\frac{e^{-H}}{\tr(e^{-H})}\Big)\S\Big(\rho\|\frac{e^{-Q^2}}{\tr(e^{-Q^2})}\Big)
\\
&\geq\S(\rho)^2-\S(\rho)\big(\log(\tr(e^{-H}))+\log(\tr(e^{-Q^2}))\big)
\\
&\qquad+\log(\tr(e^{-H}))\log(\tr(e^{-Q^2})).
\end{align*}
Here we can use Jensen's inequality:
\begin{align}\label{eq_deadend}
\tr(\rho H)\tr(\rho Q^2)
&\geq
\S(\rho)^2-\S(\rho)\big(\log(\tr(e^{-H}))+\log(\tr(e^{-Q^2}))\big)
\\
&\qquad+\big(\log(d)-\frac{1}{d}\tr(H))(\log(d)-\frac{1}{d}\tr(Q^2)\big).
\end{align}
The sum of the first two terms $\S(\rho)^2-\S(\rho)(\log(\tr(e^{-H}))+\log(\tr(e^{-Q^2}))$ is negative for big $d$, since $\log(\tr(e^{-Q^2}))$ is of order $d^2$ and $\log(\tr(e^{-H})>1$ because of the normalization of $H$.
So the sign of the right hand side depends on the eigenvalues of $H$.
Although this inequality might be non-trivial, we found no way to express the right hand side in meaningful quantities.

Also we can not expect the right hand side in \eqref{eq_deadend} to become solely dependent on $\S(\rho)$ since we haven't used $H=\frac12 P^2+V(x)$.
Because of that, we can scale the Hamiltonian on the left hand side via $\tilde{H}=\alpha H$ for small $\alpha$.
So the left hand side can become any non-negative number, without changing the state $\rho$, thus no non-trivial inequality can result directly from this.
\section{Conclusion}
In this thesis we discussed bounds of the type $\tr(\rho H)\tr(\rho Q^2)\geq f(\S(\rho))$.
First we examined the bound given in \cite{dam-nguyen}.
This bound does not hold, as we have shown with our first counterexample.
It shows an error by a factor $1/2$ in the limit for $d\to\infty$.
But since the Hamiltonian, as well as the state are most likely not optimal, the error must be assumed to be larger.
Here, $d$ is the number of degrees of freedom.

In contrast to that, the numerical results obtained in the setting of $\Cn$ suggest a bound
\begin{equation*}
\tr\left(\rho H\right)  \tr(\rho Q^2) \geq \frac{\h^2d^2}{2m}(\alpha\exp(\S(\rho)/d)-1)^2,
\end{equation*}

with a factor $\alpha\leq0.4$.
Although we did not transfer our results to the infinite-dimensional case via the thermodynamical limit, they strongly indicate the existence of such a bound for the finite-dimensional case.

As we have seen in the second example (theorem \ref{thm_ex_2}), the energy-surface bound does not hold for small $d$.
The result suggests that no similar bound holds in this scale.

\appendix
\counterwithin{theorem}{section}
\section{Appendix}
Let $\mathcal{H}$ be a finite or infinite-dimensional Hilbert space. 
\begin{lemma}\label{lem_matr_ps}
Let $f:[a,b]\to\R$, $a,b\in\R$, be an analytic function with power series coefficients $(a_k)_{k\in\N}$ around $x_0\in\R$ and convergence radius $r>0$, i.e.
\[
f(x)=\sum_{k=0}^\infty a_k (x-x_0)^k.
\]
We assume that $(a,b)\subseteq B_r(x_0)$ and that $f$ convergences absolutely on the convergence radius.
Then the function
\begin{align*}
f:[a,b]_\mathcal{H}&\to \mathcal{B}(\mathcal{H}),\\
\rho&\mapsto\sum_{k=0}^\infty a_k (\rho-x_0\1)^k
\end{align*}
is well defined and the sum converges in operator norm.
We have for the spectrum
$
\sigma(f(\rho))=f(\sigma(\rho)),
$
and the result has real trace:
$
\tr(f(\rho))\in \R\cup\infty.
$

By the same definition $f$ is a well defined function $\{\rho\in\Cn|\VERT{\rho-x_0\1}\leq r\}\to\mathcal{B}(\mathcal{H})$.
\end{lemma}
\begin{proof}
Let $\rho\in[a,b]_\mathcal{H}$. Then we have a spectral decomposition
\[
\rho=\sum_{i=1}^{\dim(\mathcal{H})}\lambda_i\ketbra{\psi_i},
\]
with $\lambda_i\in[a,b]$ and $(\ket{\psi_i})_i$ an orthonormal basis.
Now we see that $f(\rho)$ is just the application of $f$ on the eigenvalues of $\rho$:
\begin{align*}
\sum_{k=0}^\infty a_k (\rho-x_0\1)^k&=
\sum_{k=0}^\infty a_k \Big(\sum_{i=1}^{\dim(\mathcal{H})}\lambda_i\ketbra{\psi_i}-\sum_{i=1}^{\dim(\mathcal{H})}x_0\ketbra{\psi_i}\Big)^k
\\
&=\sum_{k=0}^\infty a_k\Big(\sum_{i=1}^{\dim(\mathcal{H})}(\lambda_i-x_0)\ketbra{\psi_i}\Big)^k
\\
&=\sum_{k=0}^\infty a_k\sum_{i=1}^{\dim(\mathcal{H})}(\lambda_i-x_0)^k\ketbra{\psi_i}
\\
&=\sum_{i=1}^{\dim(\mathcal{H})}f(\lambda_i)\ketbra{\psi_i},
\end{align*}
where the convergence of the i-sum is with respect to the operator norm and we used that the power series converges absolutely.
The last sum converges, since $\lambda_i\in \overline{B_{r}(x_0)}$ and $f(\overline{B_{r}(x_0)})$ is compact.
Thus $(f(\lambda_i))_i$ is bounded.
The statements about the spectrum and the trace follow immediately from the last equation.

For the definition on $\{\rho\in\Cn|\VERT{\rho-x_1\1}\leq r\}$ we have:
\begin{align*}
\left\|\sum_{k=0}^\infty a_k (\rho-x_0\1)^k\right\|
&\leq \sum_{k=0}^\infty |a_k| \VERT{\rho-x_0\1}^k
\\
&\leq \sum_{k=0}^\infty |a_k| r^k
\\
&<\infty.
\end{align*}
So $f$ is well defined on this set, too.
\end{proof}
\begin{lemma}[Von Neumann Entropy]\label{vnEntropy}
The Von Neumann Entropy
\[
S:\mathcal{S(H)}\to\overline{\R^+_0}, \rho\mapsto-\tr(\rho\log(\rho))
\]
is well defined with $\overline{\R^+_0}=\R^+_0\cup\infty$.
If $\mathcal{H}$ is finite dimensional, the same definition holds for $[a,b]_\mathcal{H}$ instead of $\mathcal{S(H)}$ and the result is always finite.
\end{lemma}
\begin{proof}
The logarithm on the real line is defined by the power series around $1$ with convergence radius $r=1$
\[
\log(x)=\sum_{k=1}^\infty\frac{(-1)^{k+1}}{k}(x-1)^k.
\]
$f:(0,2)\to\R,x\mapsto x\log(x)$ has then the power series 
\[
x\log(x)=x\sum_{k=1}^\infty\frac{(-1)^{k+1}}{k}(x-1)^k,
\]
with convergence radius $r=1$ around 1.
It can be rewritten using the transformation $x=y+1$:
\begin{align*}
x\log(x)
&=
(y+1)\sum_{k=1}^\infty\frac{(-1)^{k+1}}{k}(y)^k
\\
&=\sum_{k=1}^\infty\frac{(-1)^{k+1}}{k}y^{k+1}
+\sum_{k=1}^\infty\frac{(-1)^{k+1}}{k}y^k
\\
&=\sum_{k=2}^\infty\frac{(-1)^{k}}{k-1}y^{k}
+\sum_{k=1}^\infty\frac{(-1)^{k+1}}{k}y^k
\\
&=y+\sum_{k=2}^\infty(-1)^k y^{k} \Big(\frac{1}{k-1}-\frac{1}{k}\Big)
\\&=
y+\sum_{k=2}^\infty(-1)^k y^{k}\frac{1}{k(k-1)}.
\end{align*}
This series converges absolutely on the convergence radius $r=1$ since the it can be upper bounded by series $a_n=1/n^2$.
By lemma \ref{lem_matr_ps} we now have that $f$ is well defined on $[0,1]_\mathcal{H}$ and also on $\mathcal{S(H)}$.

We also have $f\leq0$ on $[0,1]$, so for the spectrum with respect to a state $\rho$ we get $\sigma(f(\rho))\leq0$ and  $\S(\rho)=-\tr(f(\rho))\geq0$ for all $\rho\in\mathcal{S(H)}$. This proves the remarks.
\end{proof}
\newpage
\bibliographystyle{alpha}
\bibliography{./tex/masterarb}
\end{document}